\documentclass[11pt]{article}

\usepackage[T1]{fontenc}
\usepackage{lmodern}

\usepackage[a4paper,margin=1in]{geometry}
\usepackage{microtype}
\usepackage{amsmath,amssymb,amsthm}
\usepackage{graphicx}
\usepackage{subcaption}
\usepackage{booktabs}
\usepackage{array}
\usepackage[round,authoryear]{natbib}
\usepackage{xcolor}
\usepackage[colorlinks=true,linkcolor=blue!50!black,citecolor=blue!50!black,urlcolor=blue!50!black]{hyperref}

\graphicspath{{figs/}}

\newtheorem{theorem}{Theorem}
\newtheorem{proposition}[theorem]{Proposition}
\newtheorem{lemma}[theorem]{Lemma}
\theoremstyle{remark}
\newtheorem{remark}[theorem]{Remark}

\newcommand{\bfone}{\mathbf{1}}
\newcommand{\bfalpha}{\boldsymbol{\alpha}}
\newcommand{\bft}{\mathbf{t}}

\title{The $M/PH_2/1$ Queue: Explicit Distributions,\\ Sojourn Times and Numerical Benchmarks}
\author{Yossi Luzon\thanks{School of Industrial Engineering and Management, AFEKA -- The Tel Aviv Academic College of Engineering, 38 Mivtza Kadesh Street, Tel Aviv 6998812, Israel. Email: \texttt{yossefl@afeka.ac.il}. ORCID: 0000-0003-0001-9572.}}
\date{November 2026}

\begin{document}
\maketitle

\begin{abstract}
Two-phase phase-type ($PH_2$) distributions are among the service-time models used most often in applied queueing, because they capture coefficients of variation above and below one with few parameters. This note collects, in one place and in ready-to-use form, the stationary queue-length distribution $p_n = A_1 r_1^n + A_2 r_2^n$ and the sojourn-time density $f_{W_s}(t) = c_1 e^{\alpha_1 t} + c_2 e^{\alpha_2 t}$ for the $M/E_2/1$, $M/Hypo_2/1$, $M/Hyper_2/1$ and $M/Coxian_2/1$ queues, with every coefficient written in the parameters that a fitting procedure returns. None of these results is new in substance; however, obtaining them from the general theory takes work that each user would otherwise have to repeat. The formulas are obtained from the Pollaczek--Khinchine transforms, related to the matrix-geometric and phase-type representations of the same queues, and verified numerically in 40-digit arithmetic. The note also gives the matrices needed to reproduce the results in standard software, a comparison with the matrix-analytic library BuTools, formulas for tail and threshold quantities and their derivatives, and a worked example calibrated to published surgical-time data.
\end{abstract}


\section{Introduction}\label{sec:intro}

\subsection{Two-phase service in practice}\label{sec:practice}

Phase-type distributions of order $k$ ($PH_k$) approximate any positive service-time law arbitrarily closely, but it is their two-phase members that recur most often in practice. Two phases are enough to produce squared coefficients of variation both below and above one, and few enough parameters to fit reliably from data. In call center operations, \citet{roubos2013call} found that hyperexponential-2 ($Hyper_2$) distributions fit customer patience data better than exponential or Weibull alternatives, a finding later used in delay-announcement design \citep{yu2017delay}. In healthcare, Coxian-2 distributions are a standard model for length of stay, used in geriatric care \citep{marshall2004coxian}, stroke patient flow \citep{jones2019modelling} and patient pathway analysis \citep{faddy2005markov}. Erlang-2 failure processes appear in reliability monitoring \citep{hsu2011two}, and Coxian phase-type residence times appear outside service systems as well, for example in stochastic epidemic models with regime switching \citep{neda2026stochastic}. As \citet{faddy1998inferring} observed, two-phase Coxian fits often suffice while avoiding the identifiability problems of higher orders.

Practitioners who use these models need ready-to-use formulas for queue-length distributions, sojourn times, tail probabilities and their derivatives, written in the parameters they actually fit. Collecting and verifying those formulas is the purpose of this note.

\subsection{What is known}\label{sec:known}

The stationary analysis of $M/PH/1$ is classical, and it can be reached from two directions.

In the \emph{transform} direction, the Pollaczek--Khinchine (PK) formula \citep{takacs1962introduction,cohen1982single,asmussen2003applied} gives the queue-length generating function and the sojourn-time Laplace--Stieltjes transform (LST) of any $M/G/1$ queue. When the service time is phase-type, both are rational, and partial fractions give the distributions as finite sums of geometric and exponential terms. Root-finding implementations of this idea for general phase counts are standard \citep[e.g.][]{adan1996analyzing}.

In the \emph{matrix-geometric} direction \citep{neuts1994matrix,latouche1999introduction}, the stationary vector of the level--phase process has the form $\boldsymbol{\pi}_n = \boldsymbol{\pi}_1 \mathbf{R}^{n-1}$. For $M/PH/1$ the rate matrix is available explicitly. \citet{ramaswami1986general} showed that whenever the downward transition matrix of a quasi-birth-death (QBD) process has rank one, $\mathbf{G} = \bfone\bfalpha$ and $\mathbf{R}$ is given in closed form; their Section~3 names $M/PH/1$ as the paradigm example of the continuous-time case and refers to \citet{gillent1983semi} for it. Symbolic rate matrices for $M/Hypo_k/1$ and $M/Hyper_k/1$ were also given by \citet{marin2014explicit}, and phase-level expressions for related controlled queues by \citet{wang1995optimal} and \citet{wang1999optimal}. For sojourn times, \citet[Cor.~IX.3.1]{asmussen2010ruin}, building on \citet{neuts1994matrix}, show that the $M/PH/1$ waiting time is zero-modified phase-type with an explicit representation; adding the service time keeps it phase-type.

Consequently, every distribution in this note follows from known results. The $p_n$ formulas are the spectral decomposition of a known $2\times 2$ matrix, and the sojourn-time formulas are the explicit inversion of a known two-phase representation. Section~\ref{sec:connections} makes both statements precise.

\subsection{What this note provides}\label{sec:scope}

The note provides four things.
\begin{enumerate}
\item \emph{Explicit scalar formulas in native parameters} (Sections~\ref{sec:queue}--\ref{sec:sojourn}). Each of the four $M/PH_2/1$ models is stated in its own parameters, so that the output of a fitting procedure can be used directly. Tables~\ref{tab:sojourn_reference}, \ref{tab:mg_matrices} and~\ref{tab:ph_matrices} collect the coefficients and the matrix representations.
\item \emph{A verification against both literature routes} (Section~\ref{sec:connections}). All formulas are checked against the \citet{ramaswami1986general} rate matrix and the \citet{asmussen2010ruin} sojourn-time representation in 40-digit arithmetic.
\item \emph{Numerical benchmarks} (Section~\ref{sec:butools}). The exact forms are used as an oracle for the BuTools library \citep{horvath2012butools}, which separates the behavior of its queue-length and sojourn-time computations.
\item \emph{Worked uses} (Sections~\ref{sec:applications}--\ref{sec:casestudy}). Threshold and tail formulas, sensitivity derivatives, and a calibrated case-mix example.
\end{enumerate}
Section~\ref{sec:beyond} discusses what changes for three or more phases.

\section{Framework}\label{sec:prelim}

\subsection{The Pollaczek--Khinchine formulas}\label{sec:pk}

For an $M/G/1$ queue with arrival rate $\lambda$ and service-time distribution $B$ with LST $\tilde{B}(s) = E[e^{-sB}]$ and $\rho = \lambda E[B] < 1$, the generating function of the stationary number in system is \citep{asmussen2003applied,adan2001queueing,gross2011fundamentals}
\begin{equation}\label{PK}
P(z)=\frac{(1-\rho)\tilde{B}(\lambda-\lambda z)(1-z)}{\tilde{B}(\lambda-\lambda z)-z},
\end{equation}
and the LST of the stationary FCFS sojourn time $W_s$ is \citep{gross2011fundamentals,cohen1982single}
\begin{equation}\label{eq:sojourn_LST}
\widetilde{W}_s(s) = \frac{(1-\rho)\,s\,\tilde{B}(s)}{s - \lambda + \lambda\tilde{B}(s)}.
\end{equation}
When $\tilde{B}$ is rational, both transforms are rational, and if $P(z) = \sum_i A_i/(1-r_i z)$ then $p_n = \sum_i A_i r_i^n$ \citep{wilf2005generatingfunctionology}.

\subsection{Two-phase phase-type distributions}\label{sec:ph2}

\begin{figure}[!t]
\centering
\includegraphics[width=0.9\linewidth]{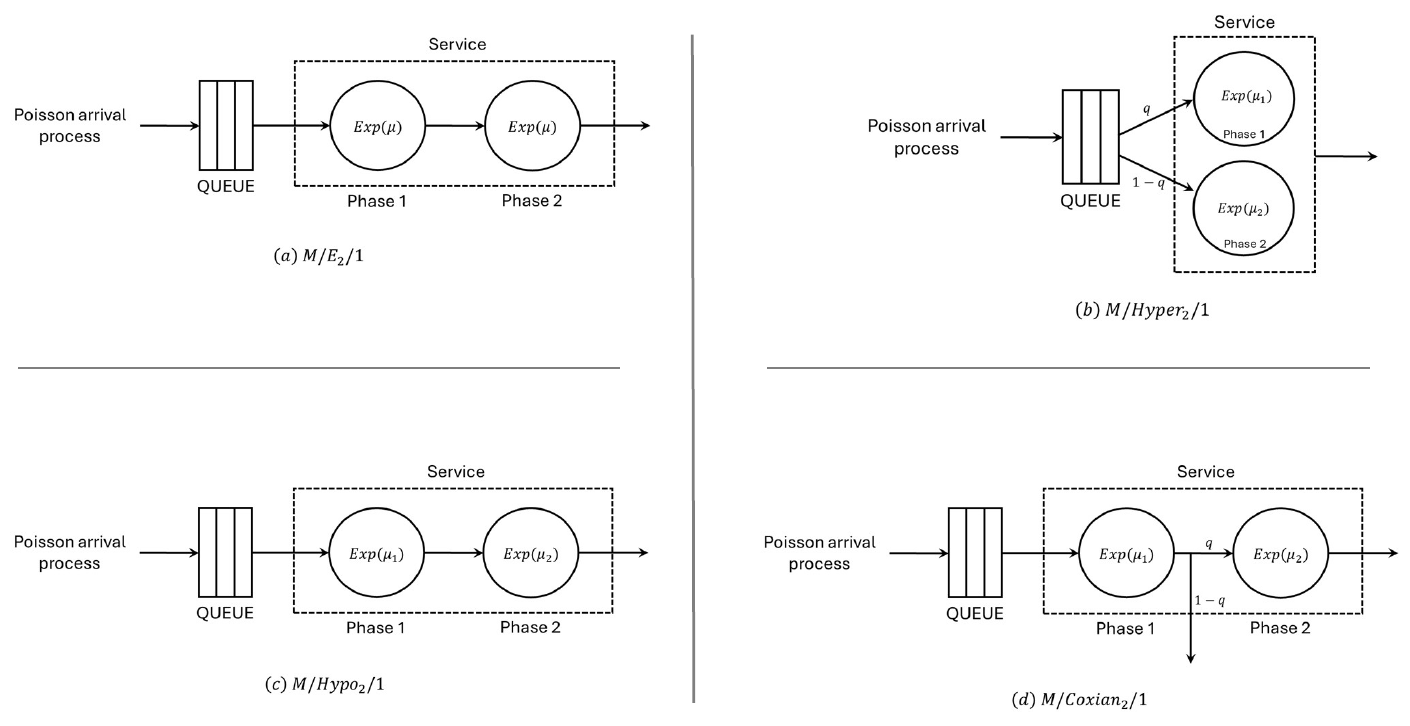}
\caption{Single-server queues with two-phase phase-type service: (a) $M/E_2/1$, (b) $M/Hyper_2/1$, (c) $M/Hypo_2/1$, (d) $M/Coxian_2/1$.}
\label{fig:systems}
\end{figure}

A phase-type distribution with representation $(\bfalpha, \mathbf{T})$ is the absorption time of a Markov process with initial vector $\bfalpha$ and sub-generator $\mathbf{T}$; its exit-rate vector is $\bft = -\mathbf{T}\bfone$ \citep{neuts1994matrix,bladt2017matrix}. The four two-phase families are: Erlang-2 ($E_2$), two identical exponential phases in series; hypoexponential-2 ($Hypo_2$), two phases in series with distinct rates; hyperexponential-2 ($Hyper_2$), a mixture of two exponentials; and Coxian-2, two phases in series with probabilistic exit after the first. Their squared coefficients of variation are $C_s^2 = 0.5$ for $E_2$, $C_s^2 \in [0.5,1)$ for $Hypo_2$, $C_s^2 > 1$ for $Hyper_2$ and $C_s^2 \in [0.5,\infty)$ for Coxian-2. The value $0.5$ is the smallest $C_s^2$ attainable with two phases. Table~\ref{tab:ph_matrices} in Section~\ref{sec:sojourn_validation} lists the representations.

\medskip
\noindent\textbf{Sign convention.} Throughout, $X_{1,2} = a \mp b$ means $X_1 = a - b$ and $X_2 = a + b$. When $\pm$ and $\mp$ appear in the same expression they are linked: upper signs belong to subscript~1, lower signs to subscript~2.

\section{Queue-length distributions}\label{sec:queue}

Each formula below is obtained by substituting the model's LST into~\eqref{PK} and decomposing the resulting rational function, whose denominator is quadratic in $z$. Section~\ref{sec:connections} shows that the same formulas follow from the explicit rate matrix of \citet{ramaswami1986general}. Derivations are in Appendix~\ref{app:proofs}.

\subsection{Erlang-2}\label{sec:ME2}

With two phases of rate $\mu$, $\tilde{B}(s) = (\mu/(\mu+s))^2$ and $\rho = 2\lambda/\mu$, and \eqref{PK} becomes
\begin{equation}\label{PH-PL}
P(z)=\frac{1-\rho}{1-\rho z-\rho^{2}z(1-z)/4}.
\end{equation}

\begin{proposition}[$M/E_2/1$]\label{Theorem_1}
For $n \geq 0$, $p_n = A_1 r_1^n + A_2 r_2^n$ with
\[
r_{1,2} = \frac{2\rho}{\rho \mp \sqrt{\rho(\rho+8)} + 4},\qquad
A_{1,2} = \pm\frac{8(1-\rho)}{\sqrt{\rho(\rho+8)}\,\big(\rho \mp \sqrt{\rho(\rho+8)} + 4\big)}.
\]
\end{proposition}

\subsection{Hypoexponential-2}\label{sec:MHypo2}

With phase rates $\mu_1 \neq \mu_2$ in series, $\tilde{B}(s)=\mu_1\mu_2/((\mu_1+s)(\mu_2+s))$ and $\rho=\lambda(1/\mu_1+1/\mu_2)$.

\begin{proposition}[$M/Hypo_2/1$]\label{thm:hypo2}
For $n \geq 0$, $p_n = A_1 r_1^n + A_2 r_2^n$ with
\[
r_{1,2} = \frac{2\lambda^2}{\lambda^2 + \lambda\mu_1 + \lambda\mu_2 \mp \lambda\beta},\qquad
A_{1,2} = \pm\frac{2\lambda\mu_1\mu_2(1-\rho)}{(\lambda^2 + \lambda\mu_1 + \lambda\mu_2 \mp \lambda\beta)\,\beta},
\]
where $\beta = \sqrt{\lambda^2 + 2\lambda\mu_1 + \mu_1^2 + 2\lambda\mu_2 - 2\mu_1\mu_2 + \mu_2^2}$.
\end{proposition}

\subsection{Hyperexponential-2}\label{sec:MH2}

With rate $\mu_1$ chosen with probability $q$ and rate $\mu_2$ otherwise, $\tilde{B}(s)=q\mu_1/(\mu_1+s)+(1-q)\mu_2/(\mu_2+s)$ and $\rho=\lambda(q/\mu_1+(1-q)/\mu_2)$.

\begin{proposition}[$M/Hyper_2/1$]\label{thm:hyper2}
For $n \geq 0$, $p_n = A_1 r_1^n + A_2 r_2^n$ with
\begin{align*}
r_{1,2} &= \frac{2\lambda^2}{\lambda^2+\lambda(\mu_1+\mu_2) \mp \beta},\\[3pt]
A_{1,2} &= \frac{(1-\rho)\lambda\Big[\mu_2\big(\lambda\mu_2-\lambda(\lambda+\mu_1)\mp\beta\big) + q(\mu_1-\mu_2)\big(\lambda(\mu_1+\mu_2-\lambda)\mp\beta\big)\Big]}{\beta\big(\beta \mp \lambda(\lambda+\mu_1+\mu_2)\big)},
\end{align*}
where $\beta = \sqrt{\lambda^2\big(\lambda^2+(\mu_1-\mu_2)^2 - 2\lambda(2q-1)(\mu_1-\mu_2)\big)}$.
\end{proposition}

\subsection{Coxian-2}\label{sec:MCox2}

With phase~1 of rate $\mu_1$ followed, with probability $q$, by phase~2 of rate $\mu_2$, $\tilde{B}(s)=\mu_1((1-q) s + \mu_2)/((s + \mu_1)(s + \mu_2))$ and $\rho=\lambda(1/\mu_1+q/\mu_2)$.

\begin{proposition}[$M/Coxian_2/1$]\label{thm:coxian2}
For $n \geq 0$, $p_n = A_1 r_1^n + A_2 r_2^n$ with
\begin{align*}
r_{1,2} &= \frac{2\lambda\mu_2}{\lambda\mu_2 + \mu_1\mu_2 + \mu_2^2 \mp \beta},\\[3pt]
A_{1,2} &= \frac{\mu_{1}\mu_{2}(1-\rho)\Big((1-q)\big(\beta \pm (\mu_{2}\lambda - \mu_{1}\mu_{2})\big) \pm (1+q)\mu_{2}^{2}\Big)}{\beta\big(\mu_{2}(\lambda+\mu_{1}+\mu_{2}) \mp \beta\big)},
\end{align*}
where $\beta = \mu_2\sqrt{\lambda^2 + (\mu_1 - \mu_2)^2 + 2\lambda(\mu_2 + \mu_1(2q-1))}$.
\end{proposition}

The Coxian-2 family contains the others. At $q = 0$ it is $M/M/1$ with rate $\mu_1$; at $q = 1$ it is $M/Hypo_2/1$; at $q=1$, $\mu_1=\mu_2$ it is $M/E_2/1$. The hyperexponential case is contained by a change of parameters: if $\mu_1 > \mu_2$, the $Hyper_2$ law with parameters $(\mu_1,\mu_2,q)$ coincides with the Coxian-2 law with the same rates and continuation probability
\begin{equation}\label{eq:hyper_cox_map}
q' = \frac{(1-q)(\mu_1-\mu_2)}{\mu_1} \in (0,1),
\end{equation}
as follows by comparing the LSTs. Proposition~\ref{thm:hyper2} is nevertheless stated separately because it uses the parameters that a $Hyper_2$ fit returns.

\subsection{Real and distinct ratios}\label{sec:real}

\begin{lemma}\label{lem:real}
For each of the four models and every stable parameter set, $\beta^2 > 0$, so $r_1$ and $r_2$ are real and distinct and the coefficients involve only rational operations and one square root.
\end{lemma}

\begin{proof}
For $M/E_2/1$, $\beta^2 = \rho(\rho+8)$. For $M/Hypo_2/1$, $\beta^2 = (\lambda + \mu_1 - \mu_2)^2 + 4\lambda\mu_2$. For $M/Hyper_2/1$, $\beta^2/\lambda^2 = (\lambda - (2q-1)(\mu_1-\mu_2))^2 + 4q(1-q)(\mu_1-\mu_2)^2$. For $M/Coxian_2/1$, $\beta^2/\mu_2^2 = (\lambda - \mu_1 + \mu_2)^2 + 4\lambda\mu_1 q$, which is positive for $q \in (0,1]$; at $q=0$ the model is $M/M/1$. Each expression is positive for $\lambda > 0$ and $0<q<1$.
\end{proof}

\section{Sojourn-time distributions}\label{sec:sojourn}

Substituting each LST into~\eqref{eq:sojourn_LST}, clearing denominators and cancelling the common factor $s$ gives a proper rational function with a quadratic denominator:
\begin{align}
\widetilde{W}^{E_2}_s(s) &= \frac{(1-\rho)\mu^2}{s^2 + (2\mu-\lambda)s + \mu(\mu-2\lambda)},\label{eq:sojourn_E2_LST}\\
\widetilde{W}^{Hypo_2}_s(s) &= \frac{(1-\rho)\mu_1\mu_2}{s^2 + (\mu_1+\mu_2-\lambda)s + \mu_1\mu_2 - \lambda(\mu_1+\mu_2)},\label{eq:sojourn_Hypo_LST}\\
\widetilde{W}^{Hyper_2}_s(s) &= \frac{(1-\rho)(s\phi + \mu_1\mu_2)}{s^2 + (\mu_1+\mu_2-\lambda)s + \mu_1\mu_2 - \lambda\psi},\label{eq:sojourn_Hyper_LST}\\
\widetilde{W}^{Coxian_2}_s(s) &= \frac{(1-\rho)\mu_1\big((1-q)s + \mu_2\big)}{s^2 + (\mu_1+\mu_2-\lambda)s + \mu_1\mu_2 - \lambda\mu_2 - \lambda q\mu_1},\label{eq:sojourn_Cox_LST}
\end{align}
with $\phi = q\mu_1 + (1-q)\mu_2$ and $\psi = (1-q)\mu_1 + q\mu_2$. Partial fractions and inversion give the following.

\begin{proposition}[Sojourn time]\label{prop:sojourn_unified}
For each of the four models with $\rho<1$,
\begin{equation}\label{eq:sojourn_density}
f_{W_s}(t) = c_1 e^{\alpha_1 t} + c_2 e^{\alpha_2 t},\qquad
P(W_s \leq t) = 1 + \frac{c_1}{\alpha_1}e^{\alpha_1 t} + \frac{c_2}{\alpha_2}e^{\alpha_2 t},\qquad t\ge 0,
\end{equation}
where $\alpha_{1,2} = (-p \pm \sqrt{\Delta})/2 < 0$ are the roots of the quadratic denominator and $c_1, c_2$ are the residues given in Table~\ref{tab:sojourn_reference}.
\end{proposition}

\begin{table}[!htbp]
\centering
\caption{Sojourn-time coefficients for $M/PH_2/1$ queues, for use in~\eqref{eq:sojourn_density}, with $\alpha_{1,2} = (-p \pm \sqrt{\Delta})/2$.}
\label{tab:sojourn_reference}
\footnotesize
\setlength{\tabcolsep}{3.5pt}
\begin{tabular}{@{} l cccc @{}}
\toprule
& $M/E_2/1$ & $M/Hypo_2/1$ & $M/Hyper_2/1$ & $M/Coxian_2/1$ \\
\midrule
$p$ & $2\mu{-}\lambda$ & $\mu_1{+}\mu_2{-}\lambda$ & $\mu_1{+}\mu_2{-}\lambda$ & $\mu_1{+}\mu_2{-}\lambda$ \\[1.2em]
$\Delta$ & $\lambda(\lambda{+}4\mu)$
  & $\begin{array}{@{}c@{}} (\mu_1{-}\mu_2)^2 + \lambda^2 \\ {}+ 2\lambda(\mu_1{+}\mu_2) \end{array}$
  & $\begin{array}{@{}c@{}} (\mu_1{-}\mu_2)^2 + \lambda^2 \\ {}+ 2\lambda(1{-}2q)(\mu_1{-}\mu_2) \end{array}$
  & $\begin{array}{@{}c@{}} (\mu_1{-}\mu_2)^2 + \lambda^2 \\ {}+ 2\lambda[\mu_2 {+} \mu_1(2q{-}1)] \end{array}$ \\[1.8em]
$c_1$ & $\dfrac{(1{-}\rho)\mu^2}{\sqrt{\Delta}}$
  & $\dfrac{(1{-}\rho)\mu_1\mu_2}{\sqrt{\Delta}}$
  & $\dfrac{(1{-}\rho)(\alpha_1\phi {+} \mu_1\mu_2)}{\sqrt{\Delta}}$
  & $\dfrac{(1{-}\rho)\mu_1\!\left((1{-}q)\alpha_1 {+} \mu_2\right)}{\sqrt{\Delta}}$ \\[1.8em]
$c_2$ & $-c_1$ & $-c_1$
  & $-\dfrac{(1{-}\rho)(\alpha_2\phi {+} \mu_1\mu_2)}{\sqrt{\Delta}}$
  & $-\dfrac{(1{-}\rho)\mu_1\!\left((1{-}q)\alpha_2 {+} \mu_2\right)}{\sqrt{\Delta}}$ \\
\bottomrule
\end{tabular}

\vspace{0.4em}
{\footnotesize $\phi = q\mu_1 + (1-q)\mu_2$. Normalization: $c_1/\alpha_1 + c_2/\alpha_2 = -1$.}
\end{table}

The sequential models have a constant numerator and therefore $c_1 = -c_2$; the mixture and branching models have a numerator linear in $s$ and asymmetric residues.

\begin{remark}[A summation identity]\label{rem:summation}
Conditioning on the number found by an arriving customer gives a second expression for the sojourn density. For $M/E_2/1$ the conditional density given $N=n$ is Erlang with $2(n+1)$ phases, so
\begin{equation}\label{eq:summation_identity}
\sum_{n=0}^{\infty} (A_1 r_1^n + A_2 r_2^n) \frac{\mu^{2(n+1)} t^{2n+1}}{(2n+1)!}\, e^{-\mu t} = c_1 e^{\alpha_1 t} + c_2 e^{\alpha_2 t}.
\end{equation}
The series of polynomial-times-exponential terms collapses to two exponentials. The identity is a useful cross-check, since an error in the $p_n$ coefficients destroys it; the analogous identities hold for the other three models with their conditional densities.
\end{remark}

\section{Connection to the matrix-geometric and ruin-theoretic results}\label{sec:connections}

This section shows how both sets of formulas follow from the literature, and reports the verification.

\subsection{Queue lengths from the explicit rate matrix}\label{sec:conn_R}

Model $M/PH/1$ with service representation $(\bfalpha,\mathbf{T})$, $\bft=-\mathbf{T}\bfone$, as a continuous-time QBD with level equal to the number in system and phase equal to the service phase. The level blocks are
\[
A_0 = \lambda I,\qquad A_1 = \mathbf{T} - \lambda I,\qquad A_2 = \bft\bfalpha,
\]
listed for the four models in Table~\ref{tab:mg_matrices}. Because $A_2 = \bft\bfalpha$ has rank one, \citet[Prop.~1 and Thm.~3]{ramaswami1986general} give $\mathbf{G} = \bfone\bfalpha$ and
\begin{equation}\label{eq:R_explicit}
\mathbf{R} = -A_0\,(A_1 + A_0\mathbf{G})^{-1} = \lambda\big(\lambda(I - \bfone\bfalpha) - \mathbf{T}\big)^{-1}.
\end{equation}
The stationary distribution then has the classical matrix-geometric form
\begin{equation}\label{eq:pn_matrix}
p_n = (1-\rho)\,\bfalpha\mathbf{R}^n\bfone,\qquad n \geq 0
\end{equation}
\citep{neuts1994matrix,latouche1999introduction}. For two phases, $\mathbf{R}$ is $2\times 2$ with eigenvalues $r_1 \neq r_2$ (Lemma~\ref{lem:real}); these are the reciprocals of the roots of the PK denominator. Writing $\mathbf{R}^n = r_1^n P_1 + r_2^n P_2$ with spectral projectors $P_i = (\mathbf{R} - r_j I)/(r_i - r_j)$ gives
\begin{equation}\label{eq:A_spectral}
p_n = A_1 r_1^n + A_2 r_2^n,\qquad A_i = (1-\rho)\,\frac{\bfalpha(\mathbf{R} - r_j I)\bfone}{r_i - r_j},\quad \{i,j\}=\{1,2\}.
\end{equation}
Propositions~\ref{Theorem_1}--\ref{thm:coxian2} are~\eqref{eq:A_spectral} written out in each model's parameters. For example, for Coxian-2, \eqref{eq:R_explicit} gives, with $D = \mu_2 + \lambda(1-q)$,
\[
\mathbf{R} = \frac{\lambda}{D}\begin{pmatrix} (\lambda+\mu_2)/\mu_1 & q \\ \lambda/\mu_1 & 1\end{pmatrix},\qquad
\operatorname{tr}\mathbf{R} = \frac{\lambda(\lambda+\mu_1+\mu_2)}{\mu_1 D},\qquad
\det\mathbf{R} = \frac{\lambda^2}{\mu_1 D},
\]
and $r_{1,2}$ are the roots of $r^2 - (\operatorname{tr}\mathbf{R})\,r + \det\mathbf{R} = 0$.

\subsection{Sojourn times from the phase-type waiting time}\label{sec:conn_AA}

By \citet[Cor.~IX.3.1]{asmussen2010ruin} and the duality between the compound Poisson ruin problem and the $M/G/1$ virtual waiting time \citep[Cor.~III.3.6]{asmussen2010ruin}, together with PASTA, the stationary FCFS waiting time of $M/PH/1$ is zero-modified phase-type: it equals $0$ with probability $1-\rho$ and otherwise has representation $(\bfalpha_+, \mathbf{T} + \bft\bfalpha_+)$, where
\begin{equation}\label{eq:alpha_plus}
\bfalpha_+ = -\lambda\,\bfalpha\mathbf{T}^{-1},\qquad \bfalpha_+\bfone = \rho.
\end{equation}
Adding an independent service time and applying the convolution rule for phase-type laws \citep[Ex.~A5.6]{asmussen2010ruin} gives a representation of order $2k$ for the sojourn time. For two phases the resulting $4\times 4$ generator has eigenvalues $\{\alpha_1,\alpha_2\}\cup\operatorname{spec}(\mathbf{T})$, and the modes belonging to $\operatorname{spec}(\mathbf{T})$ cancel in the density. The cancellation is explained by a smaller representation, which follows directly from the same results; we include a short proof.

\begin{proposition}\label{prop:sojourn_ph}
In the $M/PH/1$ FCFS queue with $\rho<1$, the stationary sojourn time is phase-type with representation $(\bfalpha,\ \mathbf{T} + \bft\bfalpha_+)$, of the same order as the service time. Its exit vector is $(1-\rho)\bft$.
\end{proposition}

\begin{proof}
Let $\mathbf{U} = \mathbf{T} + \bft\bfalpha_+$. Then $-\mathbf{U}\bfone = \bft - \bft(\bfalpha_+\bfone) = (1-\rho)\bft$. Write $\mathbf{Q} = (sI - \mathbf{T})^{-1}$ for $s>0$, so that $\bfalpha\mathbf{Q}\bft = \tilde{B}(s)$. From the resolvent identity $\mathbf{T}^{-1}\mathbf{Q} = s^{-1}(\mathbf{Q} + \mathbf{T}^{-1})$ and $\bfalpha\mathbf{T}^{-1}\bft = -1$,
\[
\bfalpha_+\mathbf{Q}\bft = -\frac{\lambda}{s}\big(\bfalpha\mathbf{Q}\bft + \bfalpha\mathbf{T}^{-1}\bft\big) = \frac{\lambda}{s}\big(1 - \tilde{B}(s)\big).
\]
By the Sherman--Morrison formula, $(sI - \mathbf{U})^{-1} = \mathbf{Q} + \mathbf{Q}\bft\,\bfalpha_+\mathbf{Q}/(1 - \bfalpha_+\mathbf{Q}\bft)$, so
\[
\bfalpha (sI - \mathbf{U})^{-1}(1-\rho)\bft = \frac{(1-\rho)\tilde{B}(s)}{1 - \lambda(1-\tilde{B}(s))/s} = \frac{(1-\rho)\,s\,\tilde{B}(s)}{s - \lambda + \lambda\tilde{B}(s)},
\]
which is~\eqref{eq:sojourn_LST}. The denominator is positive for $s>0$ when $\rho<1$, and LSTs determine distributions.
\end{proof}

For two phases, Proposition~\ref{prop:sojourn_ph} says that the poles $\alpha_{1,2}$ in Table~\ref{tab:sojourn_reference} are the eigenvalues of the $2\times 2$ matrix $\mathbf{T} + \bft\bfalpha_+$, and that $f_{W_s}(t) = (1-\rho)\,\bfalpha e^{(\mathbf{T} + \bft\bfalpha_+)t}\bft$.

\subsection{Verification}\label{sec:verification}

All formulas in Sections~\ref{sec:queue} and~\ref{sec:sojourn} were checked in 40-digit arithmetic (\texttt{mpmath}) against both routes, on 300 random parameter sets per model with $\rho$ drawn uniformly from $(0.05, 0.99)$. The checks were: the residual of the matrix quadratic $A_0 + \mathbf{R}A_1 + \mathbf{R}^2A_2 = 0$ for~\eqref{eq:R_explicit}; $p_n$ from $\boldsymbol{\pi}_1\mathbf{R}^{n-1}\bfone$, with $\boldsymbol{\pi}_1$ from the boundary equations, against Propositions~\ref{Theorem_1}--\ref{thm:coxian2} for $n \in \{1,2,3,7,20,60\}$, and $p_0 = 1-\rho$; the eigenvalues of $\mathbf{R}$ against $r_{1,2}$; the order-four convolution density and the order-two density of Proposition~\ref{prop:sojourn_ph} against Table~\ref{tab:sojourn_reference} at $t \in \{0.05, 0.7, 3, 15\}$; and the eigenvalues of $\mathbf{T}+\bft\bfalpha_+$ against $\alpha_{1,2}$. The largest relative discrepancy over all these checks was $1.1\times 10^{-38}$. The spectral formula~\eqref{eq:A_spectral} was checked separately against Proposition~\ref{thm:coxian2} on 200 further parameter sets (largest relative discrepancy $1.5\times10^{-37}$), and the transform identity in the proof of Proposition~\ref{prop:sojourn_ph} was confirmed symbolically (\texttt{SymPy}) for the Coxian-2 and $Hyper_2$ representations.

\section{Numerical benchmarks against BuTools}\label{sec:butools}

Exact formulas are useful as an oracle for numerical software. This section compares them with the matrix-analytic library BuTools \citep{horvath2012butools}. The findings describe that library, in the configuration used, and are not claims about matrix-analytic methods in general. All computations used Python~3.13 on an Intel Core Ultra~7 165U with 32~GB RAM; timings are medians of 10 runs after 2 warm-up runs unless stated otherwise.

\subsection{Matrix representations}\label{sec:qbd}

Table~\ref{tab:mg_matrices} gives the QBD blocks of Section~\ref{sec:conn_R} for the four models. BuTools computes $\mathbf{R}$ numerically by cyclic reduction \citep{bini2006numerical}, which converges quadratically. For $M/PH/1$ no iteration is needed in principle, since~\eqref{eq:R_explicit} gives $\mathbf{R}$ directly; the comparison below is of the library as it is normally used.

\begin{table}[!htbp]
\centering
\caption{QBD blocks for the four models: $A_0$ arrivals, $A_1$ within-level transitions, $A_2 = \bft\bfalpha$ service completions, $\bfalpha$ the initial phase vector.}
\label{tab:mg_matrices}
\small
\begin{tabular}{@{}lcccc@{}}
\toprule
Queue & $A_0$ & $A_1$ & $A_2$ & $\bfalpha$ \\
\midrule
$M/E_2/1$ & $\begin{bmatrix}\lambda & 0\\0 & \lambda\end{bmatrix}$ &
$\begin{bmatrix}-(\lambda+\mu) & \mu\\0 & -(\lambda+\mu)\end{bmatrix}$ &
$\begin{bmatrix}0 & 0\\\mu & 0\end{bmatrix}$ & $[1, 0]$ \\[0.8em]
$M/Hypo_2/1$ & $\begin{bmatrix}\lambda & 0\\0 & \lambda\end{bmatrix}$ &
$\begin{bmatrix}-(\lambda+\mu_1) & \mu_1\\0 & -(\lambda+\mu_2)\end{bmatrix}$ &
$\begin{bmatrix}0 & 0\\\mu_2 & 0\end{bmatrix}$ & $[1, 0]$ \\[0.8em]
$M/Hyper_2/1$ & $\begin{bmatrix}\lambda & 0\\0 & \lambda\end{bmatrix}$ &
$\begin{bmatrix}-(\lambda+\mu_1) & 0\\0 & -(\lambda+\mu_2)\end{bmatrix}$ &
$\begin{bmatrix}q\mu_1 & (1-q)\mu_1\\q\mu_2 & (1-q)\mu_2\end{bmatrix}$ & $[q, 1-q]$ \\[0.8em]
$M/Coxian_2/1$ & $\begin{bmatrix}\lambda & 0\\0 & \lambda\end{bmatrix}$ &
$\begin{bmatrix}-(\lambda+\mu_1) & q\mu_1\\0 & -(\lambda+\mu_2)\end{bmatrix}$ &
$\begin{bmatrix}(1-q)\mu_1 & 0\\\mu_2 & 0\end{bmatrix}$ & $[1, 0]$ \\
\bottomrule
\end{tabular}
\end{table}

\subsection{Queue lengths}\label{sec:validation}

Figure~\ref{fig:butools_accuracy} shows the relative difference between BuTools and Propositions~\ref{Theorem_1}--\ref{thm:coxian2}. It stays at or below about $10^{-14}$ for all four models up to $\rho = 0.99$. For $M/Hyper_2/1$, stress tests up to $C_s^2 = 10^5$ and $\rho = 0.999$ (Figure~\ref{fig:stress_tests}) keep relative errors between $10^{-14}$ and $10^{-10}$, with computation times of 1.2--3.2~ms. For queue lengths, BuTools is therefore accurate throughout the range tested.

\begin{figure}[!htbp]
\centering
\includegraphics[width=0.80\linewidth]{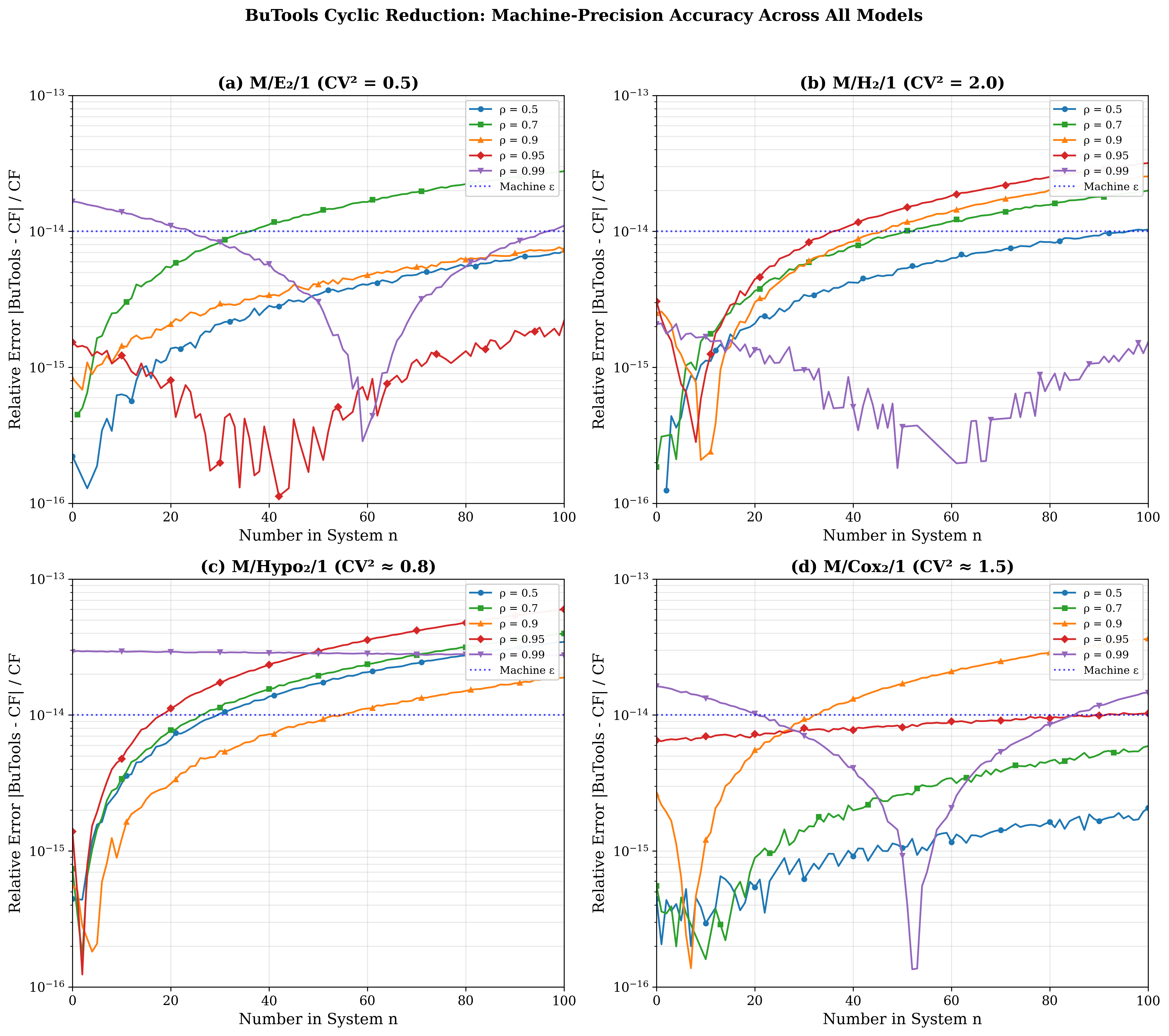}
\caption{Relative difference between BuTools (cyclic reduction) and the explicit queue-length formulas, all four models.}
\label{fig:butools_accuracy}
\end{figure}

\begin{figure}[!htbp]
\centering
\begin{subfigure}[t]{0.48\linewidth}
\centering
\includegraphics[width=\textwidth]{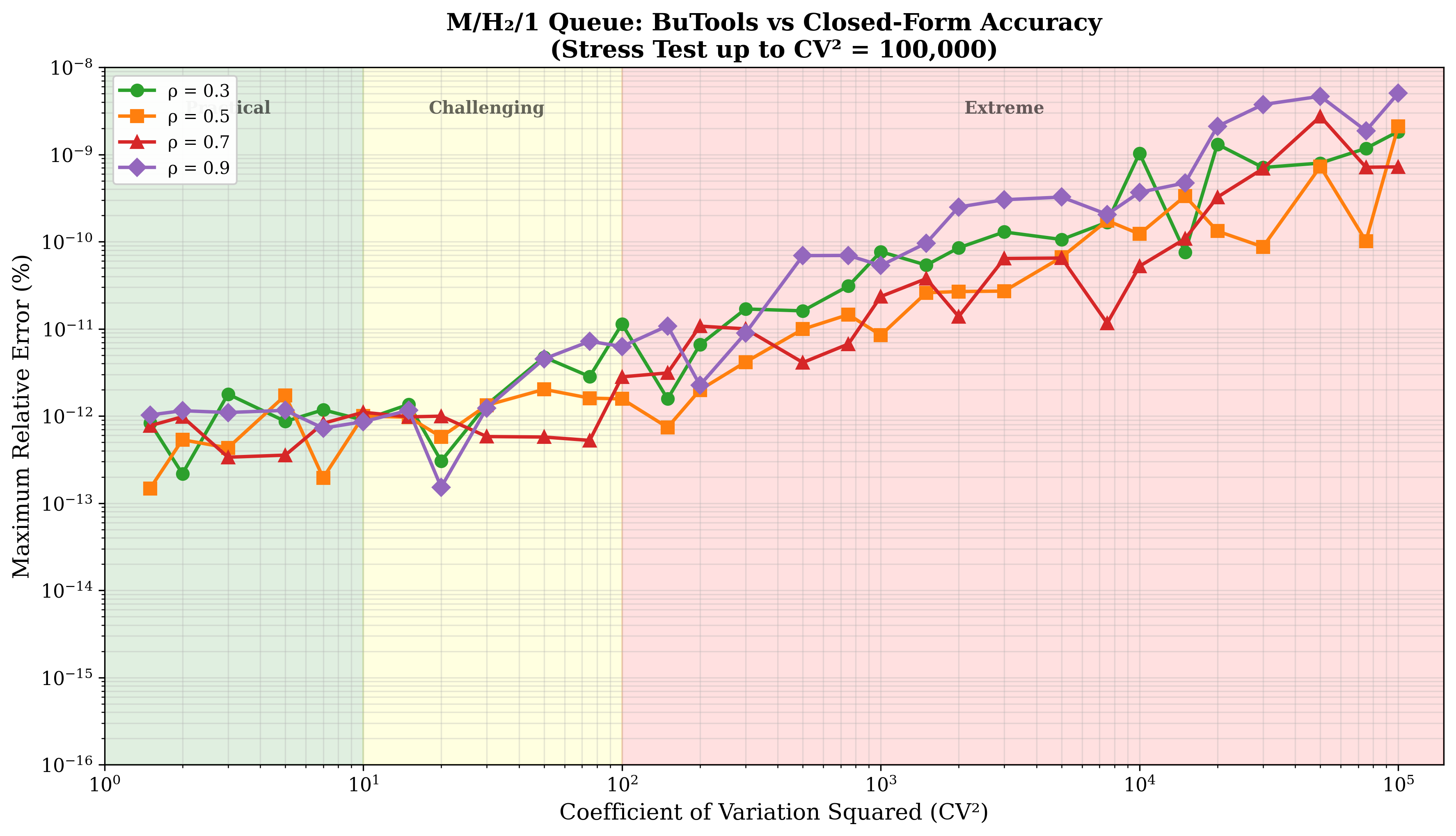}
\caption{Extreme service variability}
\end{subfigure}\hfill
\begin{subfigure}[t]{0.48\linewidth}
\centering
\includegraphics[width=\textwidth]{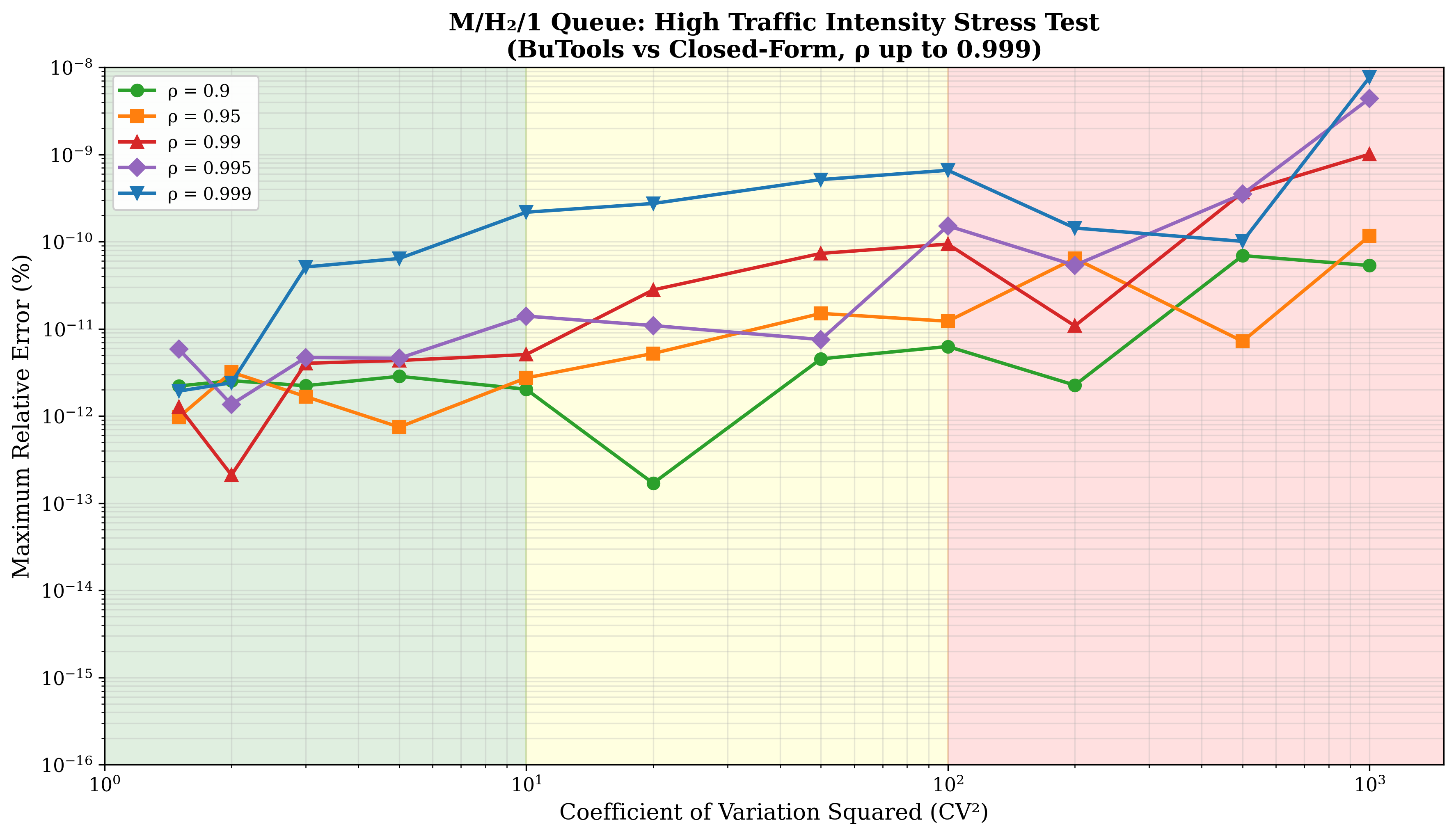}
\caption{High traffic intensity}
\end{subfigure}
\caption{Queue-length stress tests for $M/Hyper_2/1$.}
\label{fig:stress_tests}
\end{figure}

The explicit form evaluates any single $p_n$ in $O(1)$ operations, whereas evaluating $\boldsymbol{\pi}_1\mathbf{R}^{n-1}$ by repeated multiplication costs $O(n)$. Figure~\ref{fig:speedup} shows the resulting evaluation-time ratio, which grows linearly in $n$ and exceeds $10^5$ at $n = 10^4$. The ratio matters only for tail evaluation at large $n$ or for many repeated evaluations; for moderate $n$ both routes are effectively instantaneous. Repeated multiplication also underflows at very large $n$, where the explicit form does not.

\begin{figure}[!htbp]
\centering
\includegraphics[width=0.72\linewidth]{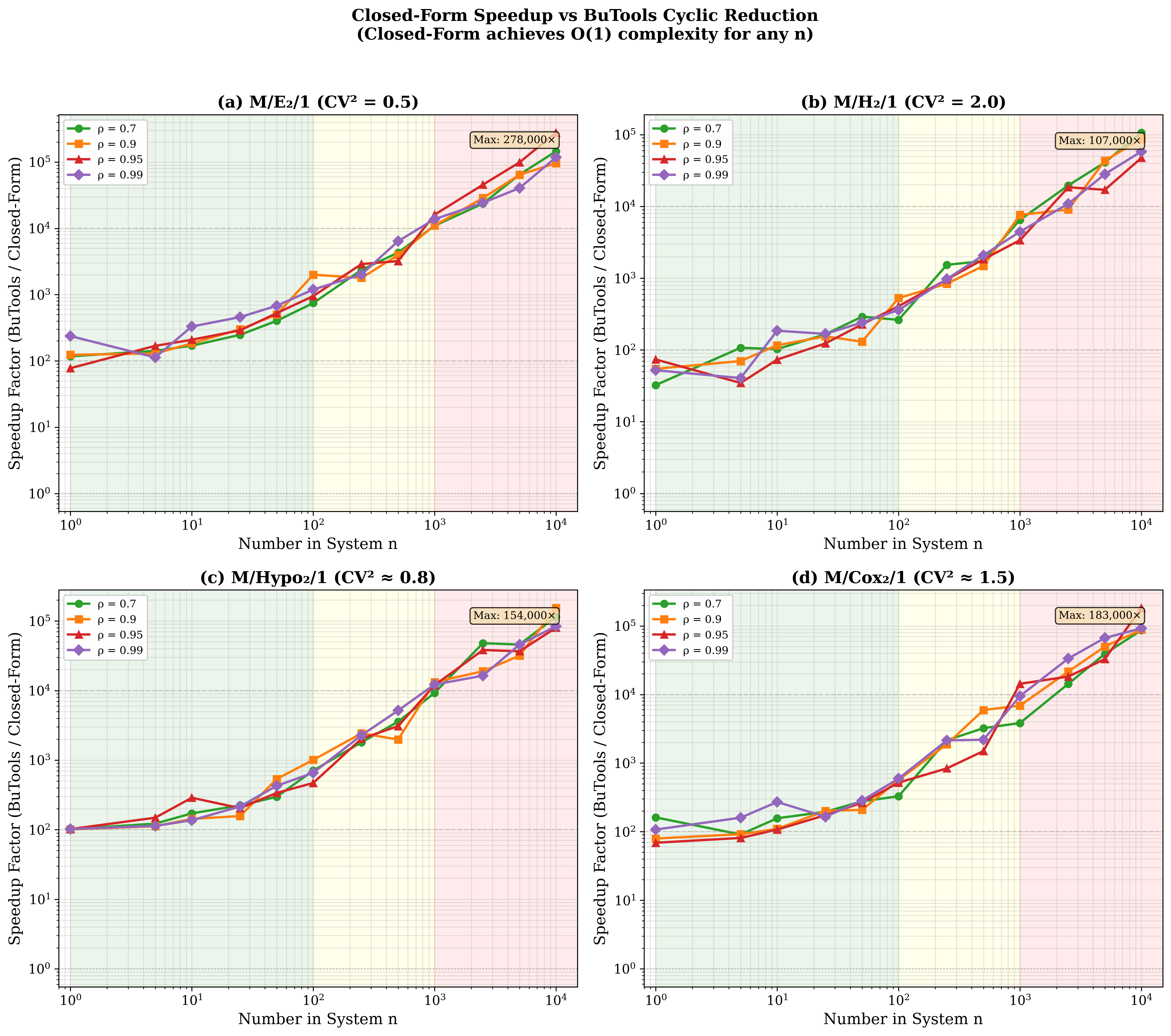}
\caption{Evaluation-time ratio, explicit formula against BuTools, as a function of $n$.}
\label{fig:speedup}
\end{figure}

\subsection{Sojourn times}\label{sec:sojourn_validation}

For sojourn times BuTools follows a different path. The function \texttt{MMAPPH1FCFS} with the \texttt{stDistr} option takes the arrival process as an MMAP with $D_0 = [-\lambda]$, $D_1 = [\lambda]$ and the service law as a phase-type representation $(\boldsymbol{\sigma}, \mathbf{S})$ (Table~\ref{tab:ph_matrices}), and evaluates the distribution through a matrix exponential. The matrix exponential is known to be sensitive when eigenvalues are close, $t$ is large, or the matrix is ill-conditioned \citep{moler2003nineteen,higham2005scaling}. We compared its output with~\eqref{eq:sojourn_density} at $t \in \{0.1, 0.5, 1, 2, 5, 10, 20, 50\}$ over the same parameter ranges as above.

\begin{table}[!htbp]
\centering
\caption{Phase-type representations $(\boldsymbol{\sigma}, \mathbf{S})$ of the service time, as passed to \texttt{MMAPPH1FCFS}.}
\label{tab:ph_matrices}
\small
\begin{tabular}{@{}lcc@{}}
\toprule
Queue & $\boldsymbol{\sigma}$ & $\mathbf{S}$ \\
\midrule
$M/E_2/1$ & $[1,\; 0]$ & $\begin{bmatrix}-\mu & \mu\\0 & -\mu\end{bmatrix}$ \\[0.8em]
$M/Hypo_2/1$ & $[1,\; 0]$ & $\begin{bmatrix}-\mu_1 & \mu_1\\0 & -\mu_2\end{bmatrix}$ \\[0.8em]
$M/Hyper_2/1$ & $[q,\; 1{-}q]$ & $\begin{bmatrix}-\mu_1 & 0\\0 & -\mu_2\end{bmatrix}$ \\[0.8em]
$M/Coxian_2/1$ & $[1,\; 0]$ & $\begin{bmatrix}-\mu_1 & q\mu_1\\0 & -\mu_2\end{bmatrix}$ \\
\bottomrule
\end{tabular}
\end{table}

Up to $\rho \approx 0.95$ all four models agree with the exact CDF to 14--15 digits (Figure~\ref{fig:sojourn_accuracy_line}). Beyond that, the error grows monotonically, reaching about $10^{-10}$ at $\rho = 0.999$ and $10^{-9}$ at $\rho = 0.9999$ for $M/E_2/1$. For $M/Hyper_2/1$ the error depends on $\rho$ and $C_s^2$ jointly (Figures~\ref{fig:sojourn_extreme} and~\ref{fig:sojourn_heatmap}). At $\rho = 0.5$ it stays near machine precision for all $C_s^2$ tested; at $\rho = 0.999$ it rises with $C_s^2$ and reaches about $10^{-6}$ at $C_s^2 = 10^4$. That is a loss of about ten significant digits relative to the queue-length computation at the same parameters. The decline is smooth, which is consistent with gradual loss of conditioning.

\begin{figure}[!htbp]
\centering
\includegraphics[width=0.62\linewidth]{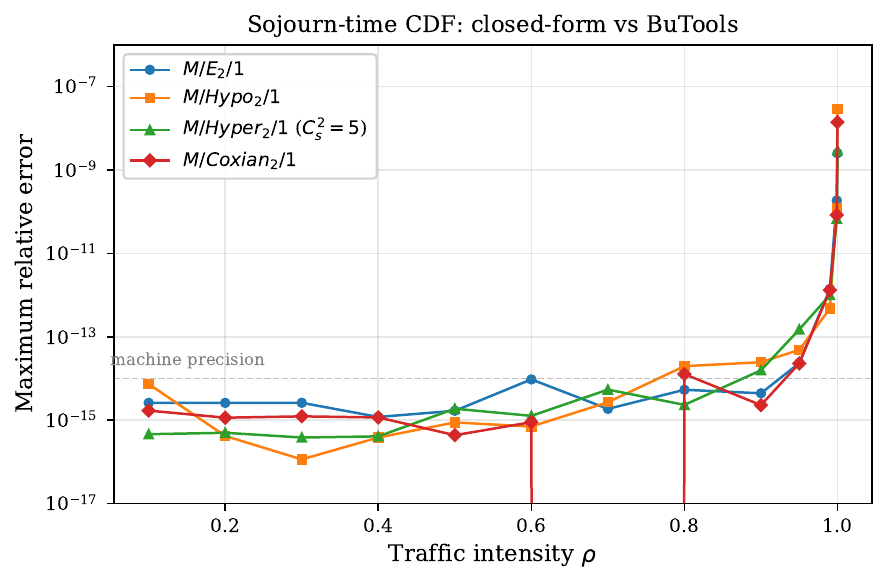}
\caption{Maximum relative error of the BuTools sojourn-time CDF against~\eqref{eq:sojourn_density}, as a function of $\rho$.}
\label{fig:sojourn_accuracy_line}
\end{figure}

\begin{figure}[!htbp]
\centering
\includegraphics[width=0.85\linewidth]{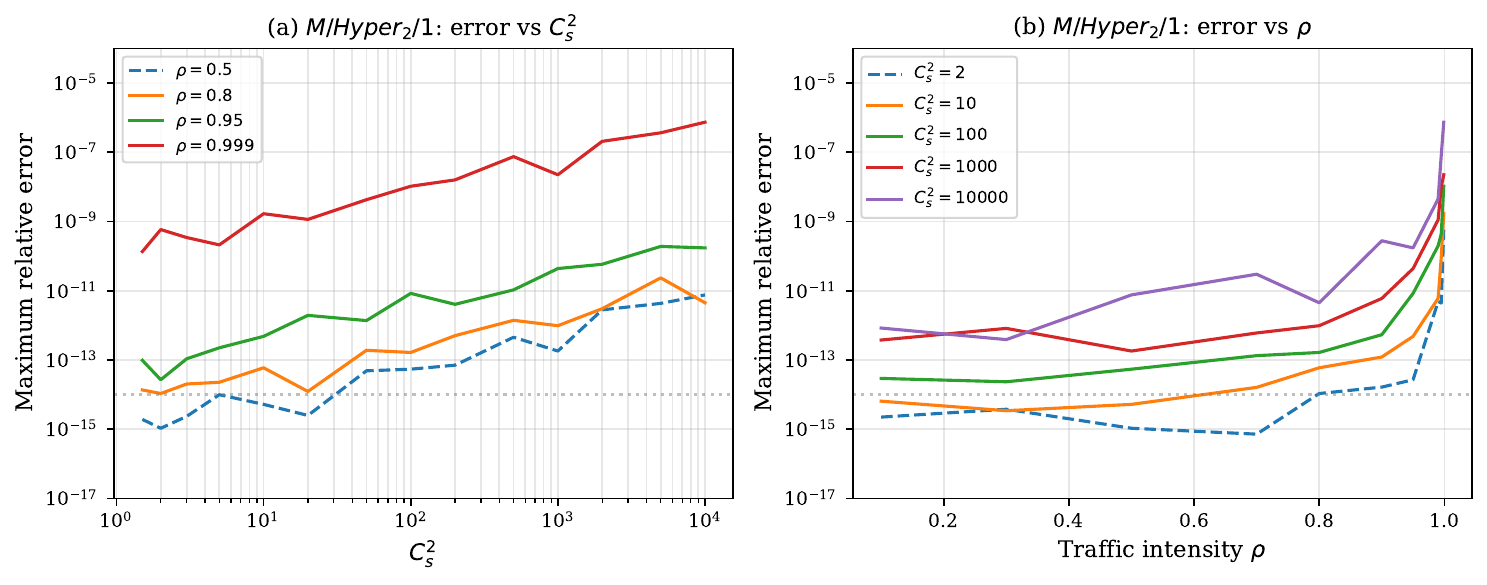}
\caption{$M/Hyper_2/1$ sojourn-time accuracy: (a) error against $C_s^2$ at fixed $\rho$; (b) error against $\rho$ at fixed $C_s^2$.}
\label{fig:sojourn_extreme}
\end{figure}

\begin{figure}[!htbp]
\centering
\includegraphics[width=0.80\linewidth]{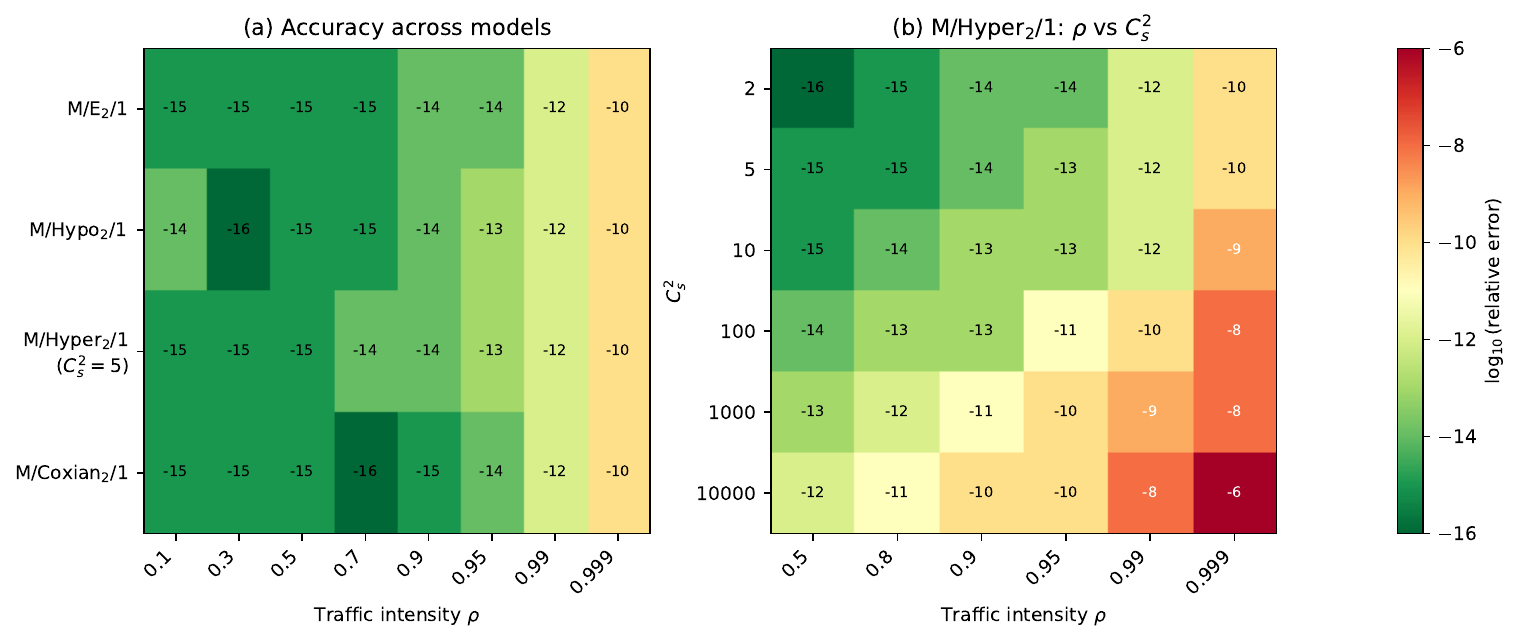}
\caption{Sojourn-time accuracy ($\log_{10}$ relative error): (a) all four models across $\rho$; (b) $M/Hyper_2/1$ across $\rho \times C_s^2$.}
\label{fig:sojourn_heatmap}
\end{figure}

Two observations bound the interpretation. First, the difficulty is not intrinsic to the distribution: by Proposition~\ref{prop:sojourn_ph} the exact sojourn law has a representation of order two, and its density is a sum of two exponentials. Second, we have not identified where in the BuTools pipeline the digits are lost, whether in the solution of the matrix equations, in the construction of the sojourn-time representation, or in the matrix exponential. A root-cause analysis, and a comparison with other software and with transform inversion \citep{abate1995numerical,shortle2007waiting}, would be needed before drawing conclusions beyond this library.

In the far tail the behavior is reversed. The computed $P(W_s > t)$ matches the exact value to machine precision at very large $t$, down to values near $10^{-300}$ for $M/E_2/1$ at $\rho = 0.5$. On timing, the explicit CDF evaluated 1{,}000 time points in 0.7~ms against 95~ms for BuTools (medians of 20 runs after 3 warm-up runs), a constant factor of about 130.

\begin{table}[!htbp]
\centering
\caption{Summary of the comparison for $k=2$.}
\label{tab:complexity_comparison}
\footnotesize
\begin{tabular}{@{}lcc@{}}
\toprule
 & BuTools & Explicit formulas \\
\midrule
Single $p_n$ & $O(n)$ & $O(1)$ \\
$P(W_s \leq t)$ at $M$ points & $O(M)$, matrix exponential per point & $O(M)$, two exponentials per point \\
Queue-length accuracy & $\sim 10^{-14}$; underflow at very large $n$ & exact up to floating point \\
Sojourn-time accuracy & $10^{-15}$ to $10^{-6}$; worst at high $\rho$ and $C_s^2$ & exact up to floating point \\
Derivatives in parameters & finite differences, or implicit differentiation & explicit \\
Phase count & any $k$ & $k=2$ (root-finding for general $k$) \\
\bottomrule
\end{tabular}
\end{table}

\section{Using the explicit forms}\label{sec:applications}

\subsection{Threshold and tail quantities}\label{sec:threshold}

From $p_n = A_1 r_1^n + A_2 r_2^n$,
\begin{align}
P(N > K) &= \frac{A_1 r_1^{K+1}}{1-r_1} + \frac{A_2 r_2^{K+1}}{1-r_2}, \label{eq:tail_exact} \\
E[(N-K)^+] &= \frac{A_1 r_1^{K+1}}{(1-r_1)^2} + \frac{A_2 r_2^{K+1}}{(1-r_2)^2}, \label{eq:overflow_exact} \\
E[N \mid N > K] &= \frac{\sum_i A_i r_i^{K+1}\big(1 + K(1-r_i)\big)/(1-r_i)^2}{\sum_i A_i r_i^{K+1}/(1-r_i)}, \label{eq:conditional_exact}
\end{align}
and $P(N > k \mid N > j)$ is the ratio of~\eqref{eq:tail_exact} at $K=k$ and $K=j$. The CDF $F_N(n) = 1 - P(N>n)$ is explicit, so the smallest buffer $n^*$ with $P(N \le n^*) \ge 1-\epsilon$ is found by bisection.

\subsection{Sensitivities}\label{sec:sensitivity}

Because $A_i$ and $r_i$ are explicit functions of the parameters, derivatives such as $\partial P(N>K)/\partial\rho$ follow by the chain rule. This is a convenience rather than a necessity. For any $M/PH_k/1$ queue, each simple root of the PK denominator can be differentiated implicitly, and derivatives of $\mathbf{R}$ solve a linear (Sylvester-type) equation obtained by differentiating the matrix quadratic. So exact gradients are available for any $k$ without closed forms.

Figure~\ref{fig:sensitivity} shows $P(N>K)$ and $\partial P(N>K)/\partial\rho$ for representative parameters of the four models. The tail ordering $M/E_2/1 < M/Hypo_2/1 < M/Coxian_2/1 < M/Hyper_2/1$ follows the ordering of $C_s^2$ in this example. At high utilization the sensitivity ordering reverses: the low-variability systems are more sensitive to $\rho$ than $M/Hyper_2/1$, whose tail is already heavy.

\begin{figure}[!htbp]
\centering
\includegraphics[width=0.80\linewidth]{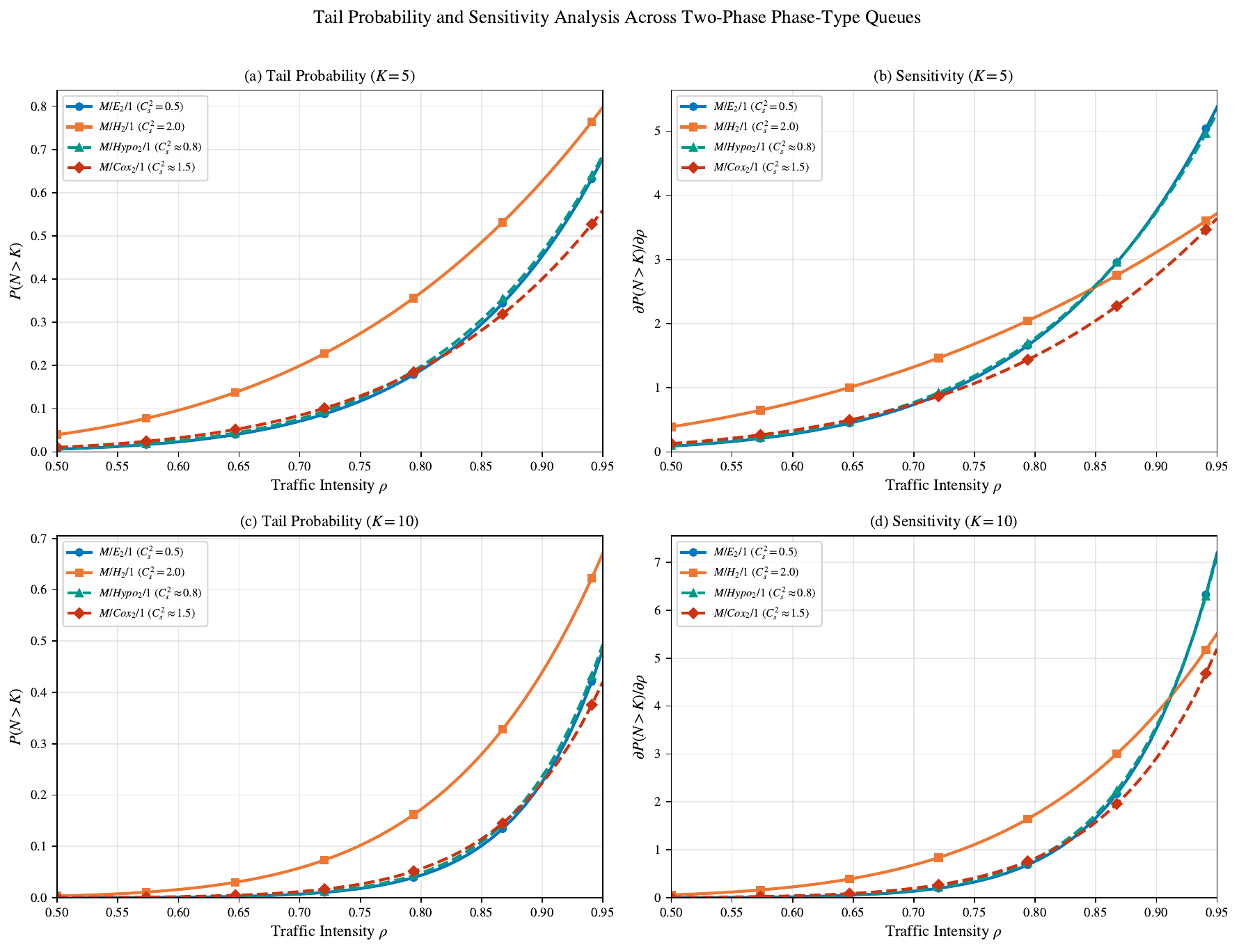}
\caption{Tail probability (left) and its derivative with respect to $\rho$ (right) for the four models; top row $K=5$, bottom row $K=10$.}
\label{fig:sensitivity}
\end{figure}

\subsection{Sojourn-time comparison}\label{sec:sojourn_comparison}

Figure~\ref{fig:sojourn_comparison} compares $M/M/1$, $M/E_2/1$ and $M/Hyper_2/1$ calibrated to the same mean sojourn time of $1$, $2$ and $10$ at $\lambda=1$. The $M/M/1$ rates are $\mu = 2$, $1.5$, $1.1$; the $M/E_2/1$ phase rates are $\mu = 2+\sqrt{3}$, $2.823$, $2.154$; and $M/Hyper_2/1$ uses $\mu_1 = 1$, $\mu_2 = 3$ with $q = 0.5(4-\sqrt{13})$, $0.445$, $0.856$. With equal means, $M/Hyper_2/1$ has more mass at short times and a heavier tail, so a target such as $P(W_s > 5) \le 0.05$ can hold for $M/E_2/1$ and fail for $M/Hyper_2/1$.

\begin{figure}[!htbp]
\centering
\includegraphics[width=0.68\linewidth]{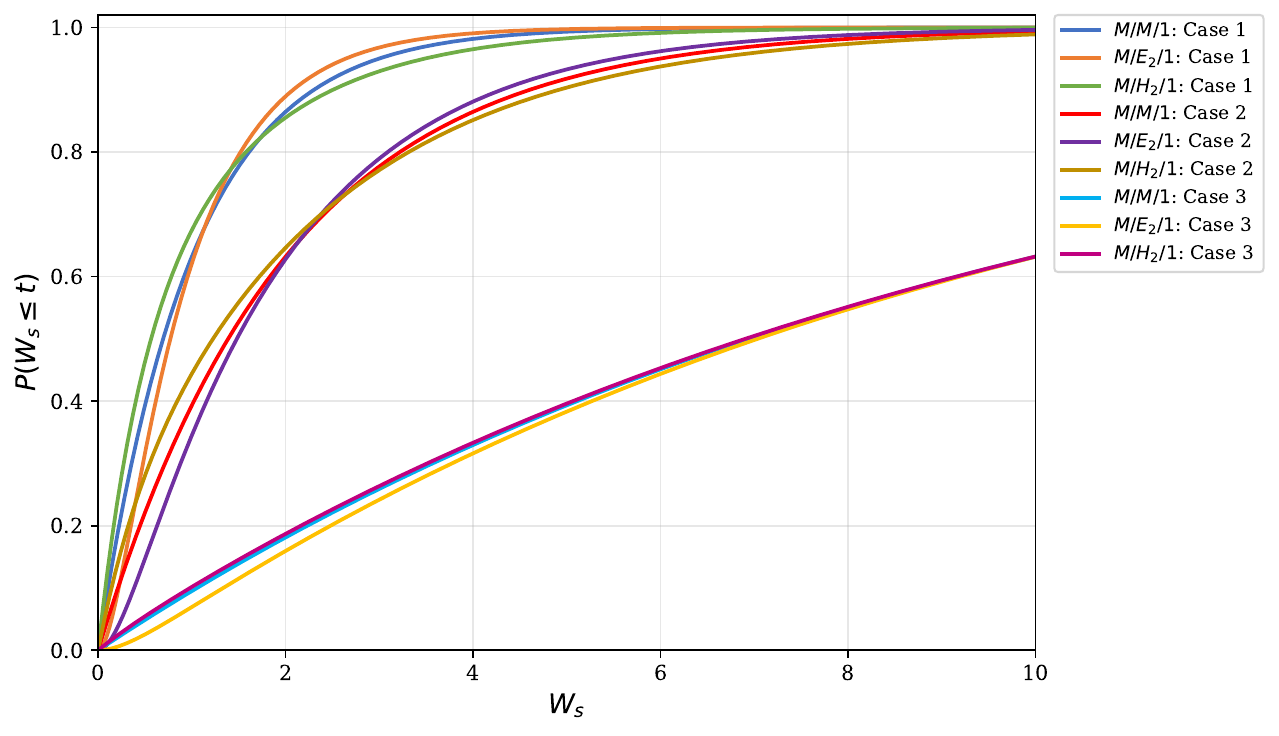}
\caption{Sojourn-time CDFs for $M/M/1$, $M/E_2/1$ and $M/Hyper_2/1$ with equal mean sojourn times.}
\label{fig:sojourn_comparison}
\end{figure}

\section{Worked example: case-mix sensitivity in a procedure suite}\label{sec:casestudy}

\subsection{Setting and calibration}\label{sec:cs_setting}

\citet{strum2000surgeon} report surgical times for 46{,}322 single-procedure cases at a large teaching hospital, by Current Procedural Terminology (CPT) code. CPT~52000, cystoscopy without biopsy, accounts for $1{,}521$ cases at $18.7 \pm 17.8$ minutes; CPT~52204, cystoscopy with biopsy, for $369$ cases at $37.2 \pm 22.9$ minutes. A urological endoscopy room serving this caseload handles one patient at a time; every patient is examined, and a minority also receive a biopsy. This is the $M/Coxian_2/1$ structure, with the two phases corresponding to the two clinical components.

Surgical times are conventionally modeled as lognormal \citep{may2000fitting,spangler2004estimating}; the Coxian-2 model is a two-moment surrogate with a rational transform. The case counts give $q = 369/1{,}890 = 0.195$; at a coarser level, \citet{strum2003estimating} report 10{,}740 two-code surgeries against 46{,}322 single-code ones, a rate of $0.188$. Taking $1/\mu_1 = 18.70$ and $1/\mu_2 = 37.2 - 18.7 = 18.50$ minutes gives $E[S] = 22.31$ minutes and $C_s^2 = 0.945$, against an empirical pooled mean of $22.31$ minutes and $C_s^2 = 0.826$ from the two reported class moments.

\begin{remark}\label{rem:branching}
The Coxian-2 model assumes the biopsy decision is made during the examination. If the procedure type is fixed at booking, a mixture is the appropriate model, and the same class means and $q$ give $M/Hyper_2/1$ with $C_s^2 = 1.216$. The data do not distinguish the two readings, and the analysis below applies to either.
\end{remark}

Arrival rates are not reported, so results are given across $\rho \in [0.55, 0.97]$, that is, $1.5$ to $2.6$ cases per hour at the calibrated $E[S]$.

\subsection{Elasticity of overflow risk}\label{sec:cs_sensitivity}

Let $K$ be the backlog beyond which cases are deferred. With $\lambda$ held fixed, a change in the biopsy share $q$ changes both the shape of the service law and the load, since $\rho = \lambda(1/\mu_1 + q/\mu_2)$. The elasticity
\begin{equation}\label{eq:elasticity}
\varepsilon(K) = \frac{\partial P(N>K)}{\partial q}\cdot\frac{q}{P(N>K)}
\end{equation}
measures the proportional change in overflow risk per proportional change in $q$. Table~\ref{tab:cs_sensitivity} reports it, computed from~\eqref{eq:tail_exact} and Proposition~\ref{thm:coxian2}.

\begin{table}[!htbp]
\centering
\caption{Overflow probability and its case-mix elasticity ($1/\mu_1 = 18.70$ min, $1/\mu_2 = 18.50$ min, $q = 0.195$).}
\label{tab:cs_sensitivity}
\footnotesize
\setlength{\tabcolsep}{4pt}
\begin{tabular}{@{}l ccc ccc ccc@{}}
\toprule
& \multicolumn{3}{c}{$\rho = 0.70$} & \multicolumn{3}{c}{$\rho = 0.85$} & \multicolumn{3}{c}{$\rho = 0.95$} \\
\cmidrule(lr){2-4}\cmidrule(lr){5-7}\cmidrule(lr){8-10}
$K$ & 5 & 10 & 20 & 5 & 10 & 20 & 5 & 10 & 20 \\
\midrule
$P(N>K)$                     & 0.113 & 0.0179 & 0.000454 & 0.370 & 0.160 & 0.0301 & 0.731 & 0.561 & 0.331 \\
$\partial P(N>K)/\partial q$ & 0.529 & 0.153  & 0.00735  & 1.81  & 1.44  & 0.515  & 3.67  & 5.18  & 5.84  \\
$\varepsilon(K)$             & 0.917 & 1.666  & 3.163    & 0.958 & 1.752 & 3.342  & 0.981 & 1.801 & 3.443 \\
\bottomrule
\end{tabular}
\end{table}

The overflow probability ranges over five orders of magnitude across utilization, while at each threshold the elasticity changes by less than ten percent between $\rho = 0.70$ and $0.95$ (Figure~\ref{fig:casestudy}).

\begin{figure}[!htbp]
\centering
\includegraphics[width=0.85\linewidth]{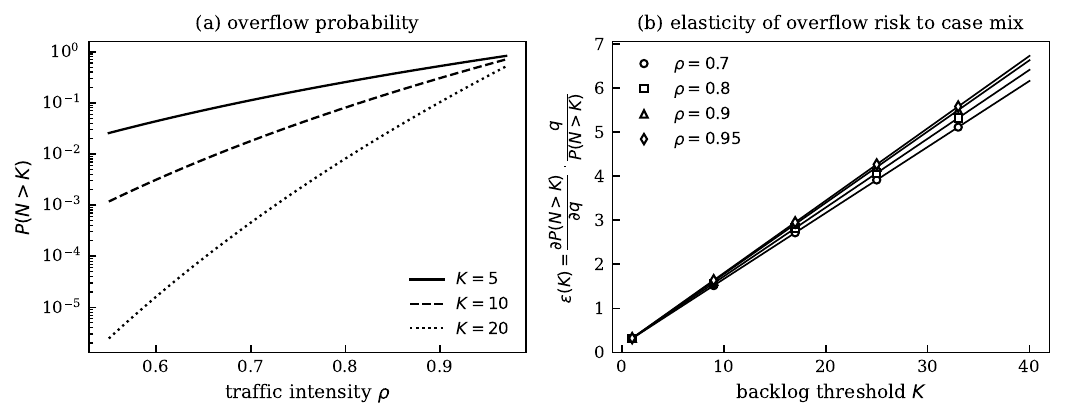}
\caption{Calibrated cystoscopy suite: (a) $P(N>K)$ against $\rho$; (b) elasticity of overflow risk with respect to the biopsy share, against $K$.}
\label{fig:casestudy}
\end{figure}

\subsection{An affine form for the elasticity}\label{sec:cs_affine}

The pattern has a simple source. Because $r_1 > r_2$, the first term of~\eqref{eq:tail_exact} dominates for moderate and large $K$, so
\begin{equation}\label{eq:logtail}
\ln P(N>K) \simeq \ln\frac{A_1}{1-r_1} + (K+1)\ln r_1,
\end{equation}
and differentiating with respect to $\ln q$ gives
\begin{equation}\label{eq:affine}
\varepsilon(K) \simeq q\,\frac{\partial}{\partial q}\ln\frac{A_1}{1-r_1} + (K+1)\,q\,\frac{\partial \ln r_1}{\partial q}.
\end{equation}
The elasticity is therefore approximately affine in $K$, with slope equal to the elasticity of the decay rate $r_1$ with respect to $q$. In this example, the slope $q\,\partial\ln r_1/\partial q$ matches the slope of the exact elasticity in $K$ to six decimal places: $0.149727$, $0.156110$, $0.161641$ and $0.164141$ at $\rho = 0.70$, $0.80$, $0.90$ and $0.95$.

This argument uses only the existence of a dominant geometric term. It therefore applies to any $M/PH_k/1$ queue whose PK denominator has a simple dominant root, and in that case the slope can be obtained for any $k$ by implicit differentiation of the root. What the two-phase forms add is that $r_1$, and hence the slope, are explicit in the parameters. Two readings follow for this suite. The proportional effect of a case-mix shift is governed mainly by the backlog threshold, since the slope varies by under ten percent across the stable range; it is the absolute risk that rises with utilization. And case mix acts through two channels, a level effect through $A_1/(1-r_1)$ and a decay-rate effect through $r_1$, of which only the second grows with $K$.

For a manager, consider a suite at $\rho = 0.85$ with the target $P(N > 20) \le 0.05$. Raising $q$ by ten percent, from $0.195$ to $0.215$, raises $P(N>20)$ from $0.0301$ to $0.0419$, which breaches the target. The same shift raises $P(N>5)$ from $0.370$ to $0.407$, a ten percent rather than a thirty-nine percent increase. Because the elasticity is local, the first-order prediction $\varepsilon(K)\times 10\%$ understates these exact changes by about four percent at $K=5$ and fifteen percent at $K=20$; for a specific displacement, direct re-evaluation is preferable.

If derivatives are instead approximated by central differences, the step size matters. Using the explicit form as a noise-free reference at $\rho = 0.85$, $K = 10$, the relative error is $6\times10^{-11}$ at $h = 10^{-6}$, $1\times10^{-3}$ at $h = 10^{-2}$ and $2\times10^{-1}$ at $h = 10^{-15}$.

\section{Beyond two phases}\label{sec:beyond}

For $M/PH_k/1$ the PK denominator, after removing the root $z=1$, is a polynomial of degree $k$, and its roots are the reciprocals of the eigenvalues of $\mathbf{R}$. The reason two phases is special is therefore elementary: at $k=2$ the quadratic formula expresses every coefficient with rational operations and a single square root, and Lemma~\ref{lem:real} shows the roots are real. Beyond $k=2$ the situation changes in presentation, not in computability.

\begin{itemize}
\item For $M/E_3/1$, the cubic $Q(z) = \frac{\rho^3}{27}z^3 - \big(\frac{2\rho^3}{27} + \frac{\rho^2}{3}\big)z^2 + \big(\frac{\rho^3}{27} + \frac{\rho^2}{3} + \rho\big)z - 1$ has discriminant
\begin{equation}\label{eq:E3_discriminant}
\Delta_{E_3} = -\frac{\rho^8(\rho^2 + 14\rho + 81)}{3^6} < 0,
\end{equation}
so it has one real root and a complex-conjugate pair at every $\rho \in (0,1)$. The distribution is still real-valued, with a damped oscillatory term $|C|\,s^n\cos(n\theta + \varphi)$ alongside a geometric one, but it is not a sum of real geometric terms. Written out with Cardano's formula, the coefficient $A_1$ alone runs to more than 15{,}000 characters.
\item For $M/Hyper_3/1$, in 429 stable configurations tested numerically, the cubic had three real roots. Expressing them by radicals then requires complex intermediates (the \emph{casus irreducibilis}), or trigonometric functions. For $M/E_4/1$, the quartic had two real and two complex roots at all 200 values of $\rho$ tested. Cubics and quartics are solvable in radicals, so explicit expressions exist for $k = 3$ and $4$, but they are long.
\item For $k \geq 5$ there is no general formula in radicals \citep{stewart2015galois}; whether particular $M/PH_k/1$ families have solvable Galois groups has not been examined here.
\end{itemize}

None of this limits computation. For any $k$, $p_n = \sum_i A_i r_i^n$ with the $r_i$ obtained by numerical root-finding and each $A_i$ an explicit residue in its root \citep{adan1996analyzing}, or $p_n = (1-\rho)\bfalpha\mathbf{R}^n\bfone$ with $\mathbf{R}$ from~\eqref{eq:R_explicit}; the sojourn time is $PH(\bfalpha, \mathbf{T} + \bft\bfalpha_+)$ by Proposition~\ref{prop:sojourn_ph}. The two-phase case is simply the largest one in which all of these quantities fit in short formulas.

\section{Summary}\label{sec:conclusions}

This note collects explicit queue-length and sojourn-time distributions for the four $M/PH_2/1$ models in their native parameters (Propositions~\ref{Theorem_1}--\ref{prop:sojourn_unified}, Table~\ref{tab:sojourn_reference}). It shows that they are explicit evaluations of the rate matrix of \citet{ramaswami1986general} and of the phase-type waiting-time result in \citet{asmussen2010ruin}, and verifies them along both routes. Used as an oracle, the formulas show that BuTools computes queue lengths at machine precision across the range tested, while its sojourn-time output loses up to about ten digits when high $\rho$ and high $C_s^2$ occur together; the cause within the library is not identified here. In a worked example calibrated to published surgical-time data, the elasticity of overflow risk to case mix is approximately affine in the backlog threshold, with slope given by the elasticity of the dominant decay rate.

\section*{Code availability}

The verification scripts of Section~\ref{sec:verification} are included as ancillary files with this arXiv submission (\texttt{anc/}). Code reproducing the BuTools benchmarks is available from the author on request. BuTools is available at \url{https://webspn.hit.bme.hu/~telek/tools/butools}.

\appendix

\section{Derivations of Propositions~\ref{Theorem_1}--\ref{thm:coxian2}}\label{app:proofs}

In each case $P(z)$ in~\eqref{PK} reduces to $(1-\rho)$ or a linear polynomial over a quadratic in $z$. Factoring the quadratic as $c(1-r_1z)(1-r_2z)$ and decomposing gives $P(z) = A_1/(1-r_1z) + A_2/(1-r_2z)$, and since $0<r_i<1$ (Lemma~\ref{lem:real}), expansion gives $p_n = A_1r_1^n + A_2r_2^n$ for $n \ge 0$, with $A_1 + A_2 = P(0) = 1-\rho$ and $A_1/(1-r_1) + A_2/(1-r_2) = P(1) = 1$.

\paragraph{$M/E_2/1$.} The denominator of~\eqref{PH-PL} is quadratic in $z$; the stated $r_{1,2}$ and $A_{1,2}$ follow after the simplification $\rho^{3/2} - \sqrt{\rho} = \sqrt{\rho}(\rho - 1)$.

\paragraph{$M/Hypo_2/1$.} Substitution gives $P(z) = \mu_1\mu_2(1-\rho)/\big(\lambda^2z^2 - \lambda(\lambda + \mu_1 + \mu_2)z + \mu_1\mu_2\big)$, whose denominator has discriminant $\lambda^2\beta^2$.

\paragraph{$M/Hyper_2/1$ and $M/Coxian_2/1$.} Substitution gives a linear numerator over a quadratic denominator, with discriminants $\beta^2$ as stated; partial fractions give the coefficients.

\medskip
\noindent Alternatively, all four follow from~\eqref{eq:R_explicit} and~\eqref{eq:A_spectral}, as verified in Section~\ref{sec:verification}.

\bibliographystyle{plainnat}
\bibliography{refs}

\end{document}